\documentclass[journal]{IEEEtran}
\usepackage{amsmath,amssymb,mathtools,amsfonts,mathrsfs,url,color,xcolor,colortbl,cite,graphicx,epstopdf,epsfig,tikz,pgfplots,lettrine,bigints,multirow}

\ifCLASSINFOpdf

\else

\fi

\newcommand\myeq{\stackrel{\mathclap{\normalfont\mbox{(a)}}}{=}}
\newcommand\myeqb{\stackrel{\mathclap{\normalfont\mbox{(b)}}}{=}}
\newtheorem{theorem}{\textbf{Theorem}}

\newtheorem{lem}{\textbf{Lemma}}

\newtheorem{col}{\textbf{Corollary}}

\newenvironment{proof}[1][Proof]{\begin{trivlist}
\item[\hskip \labelsep {\bfseries #1}]}{\end{trivlist}}
\begin{document}

\title{Handover Rate Characterization in 3D Ultra-Dense Heterogeneous Networks}

\author{\IEEEauthorblockN{Rabe Arshad$^1$, Hesham ElSawy$^2$, Lutz Lampe$^1$, and Md. Jahangir Hossain$^1$\\
       $^1$University of British Columbia, Canada.
       $^2$King Abdullah University of Science and Technology, Saudi Arabia.\vspace{-0.6cm}}

}
\maketitle
\begin{abstract}
Ultra-dense networks (UDNs) envision the massive deployment of heterogenous base stations (BSs) to meet the desired traffic demands. Furthermore, UDNs are expected to support the diverse devices e.g., personal mobile devices and unmanned ariel vehicles. User mobility and the resulting excessive changes in user to BS associations in such highly dense networks may however nullify the capacity gains foreseen through BS densification. Thus there exists a need to quantify the effect of user mobility in UDNs. In this article, we consider a three-dimensional $N$-tier downlink network and determine the association probabilities and inter/intra tier handover rates using tools from stochastic geometry. In particular, we incorporate user and BSs' antenna heights into the mathematical analysis and study the impact of user height on the association and handover rate. The numerical trends show that the intra-tier handovers are dominant for the tiers with shortest relative elevation w.r.t. the user and this dominance is more prominent when there exists a high discrepancy among the tiers' heights. However, biasing can be employed to balance the handover load among the network tiers.

\end{abstract}


\IEEEpeerreviewmaketitle
\vspace*{-0.6cm}
\section{Introduction}
Extreme densification of base stations (BSs) realizing ultra-dense networks (UDNs) is considered a key enabler to meet the spectral efficiency requirements of fifth generation (5G) cellular systems. UDNs face various challenges in supporting user mobility while offering enhanced user capacity. The deployment of more BSs within the same geographical region shrinks the service area of each BS. Thus user mobility in such a highly dense network may result in frequent user-to-BS association changes, which may jeopardize the key performance indicators (KPIs) \cite{5GUDN}. Several studies including \cite{sun2017impact} and \cite{park2016user} have been conducted in the literature to capture/manage the effect of user mobility on UDN performance metrics. However, none of the aforementioned studies quantified the effect of user mobility on user-to-BS associations. \\

Handover (HO) is the process of changing user association from one BS to another, which is triggered to maintain the best connectivity or signal-to-interference-plus-noise-ratio (SINR). HO frequency/rate depends on various factors including BS intensity, BS transmit power, and user velocity. Several researchers have characterized HO rates in wireless networks by exploiting tools from stochastic geometry. In contrast to the classical works involving coverage oriented BSs deployment that implies hexagonal coverage areas, stochastic geometry has enabled the characterization of HO rates in capacity oriented networks encompassing irregular BSs coverage regions. For instance,~\cite{Lin} exploits stochastic geometry tools to characterize HO rate in a Poisson point process (PPP) based single tier network with the random waypoint user mobility model. The HO rates for multi-tier networks are characterized in~\cite{HOben} with an arbitrary user mobility model. The model in \cite{HOben} is extended in \cite{HO_PCP} for the BSs that are deployed according to a Poisson cluster process (PCP). The authors in~\cite{diff_pathloss} conducted the HO rate analysis with different path loss exponents for each tier. However, none of the aforementioned studies incorporated user/BSs antenna heights into the mathematical analysis. A recent study \cite{lower_height} incorporated user height in the network analysis and found that decrease in the absolute height difference between the user and the BS in UDN leads to increased area spectral efficiency. This dictates that the network elements' heights should be carefully incorporated in the UDN performance analysis. To the best of authors' knowledge, no study exists that quantifies the interplay between user/BS antenna heights, association probabilities, and HO rates as a function of BS intensity, which is the main contribution in this work. In particular, we present a height-aware analytical model that characterizes the association probabilities and HO rates in a three-dimensional PPP based $N$-tier UDN. In the developed model, a user could be a pedestrian, a land vehicle or an unmanned aerial vehicle (UAV).

We consider an $N$-tier UDN where the BSs belonging to $k$-th tier, $k\in{1,...,N}$ are modeled via a homogenous PPP $\Phi_{k}$ with intensity $\lambda_k$, transmission power $P_k$, bias factor $B_k$, and antenna height $h_k$. As in ~\cite{HOben},\cite{HO_PCP}, we consider a power law path loss model with path loss exponent $\eta>2$. The disparity in the BSs transmit powers and heights yields a weighted Voronoi tessellation~\cite{voronoi}. A mobile user following an arbitrary horizontal mobility model with velocity $v$, changes its association as soon as it crosses the coverage boundary of serving BS.

 In the next section, we characterize the user-to-BS association probabilities and the service distance distributions, which are utilized to eventually derive the HO rates.
\begin{figure}[!t]
\centering
\includegraphics[width=0.85\linewidth,keepaspectratio]{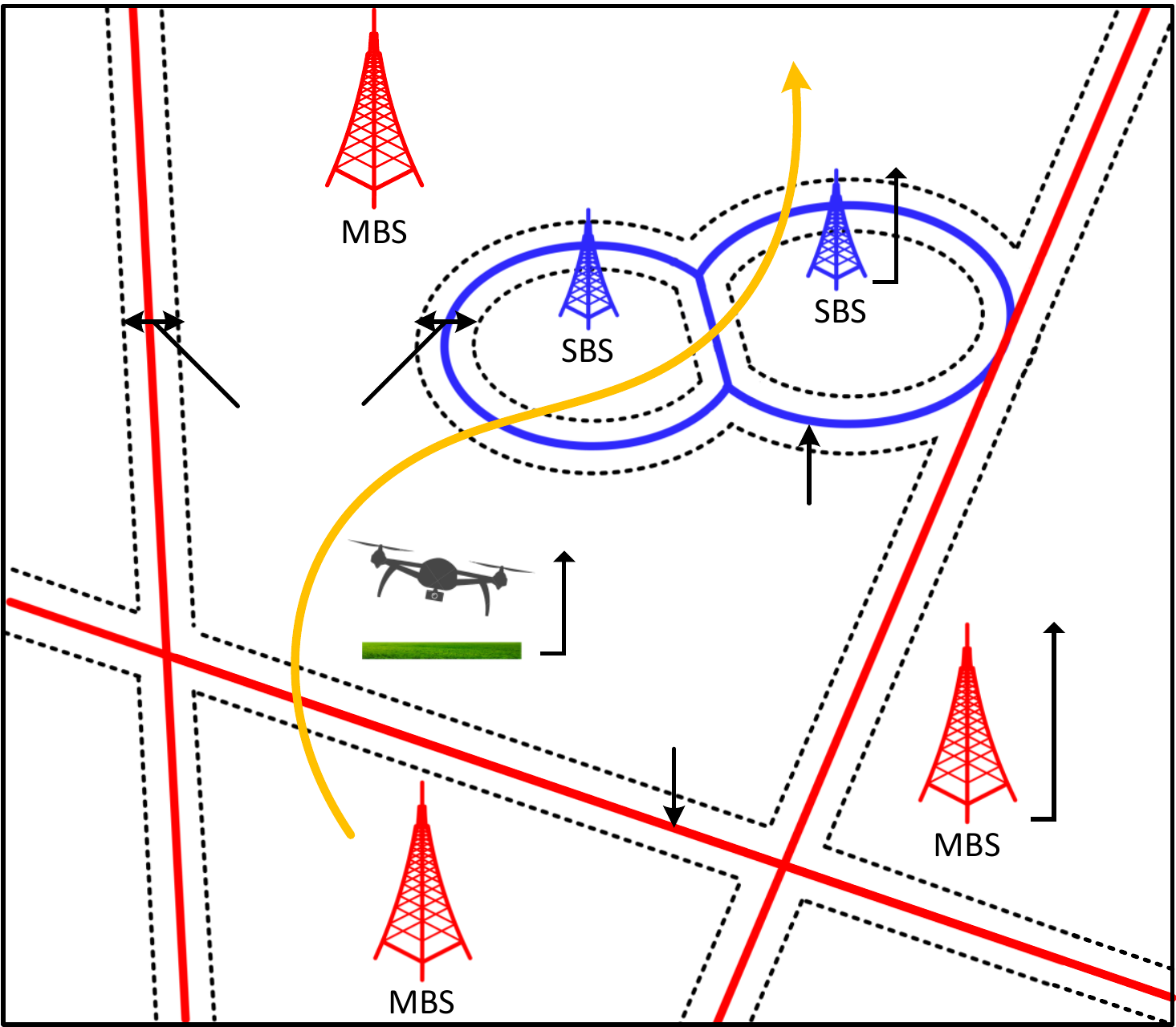}
\begin{picture}(0,0)\small
\put(-172,113){2$\Delta d$}
\put(-81,91){$\mathcal{T_\mathrm{ms}}$}
\put(-105,56){$\mathcal{T_\mathrm{mm}}$}
\put(-88,168){$\mathcal{T_\mathrm{u}}$}
\put(-59,161){${h_\mathrm{s}}$}
\put(-23,56){${h_\mathrm{m}}$}
\put(-113,76){${h_\mathrm{u}}$}
\end{picture}\normalsize
\caption{Two tier Voronoi tessellation showing inter/intra tier HO boundaries between MBS (red) and SBS (blue). Red and blue solid lines represent MBS and SBS coverage boundaries, respectively, while black dotted lines represent $\Delta d$ extended boundaries. $\mathcal{T}_{kj},k,j\in\{\mathrm{m,s}\}$ represents the HO boundaries between $k$ and $j$ tier BSs, $\mathcal{T}_\mathrm{u}$ represents the user trajectory, and $h_x, x \in \{$m, s, u$\}$ represents the height of MBS, SBS, and user, respectively }
\label{boundaries}
\vspace{-0.25cm}
\end{figure}
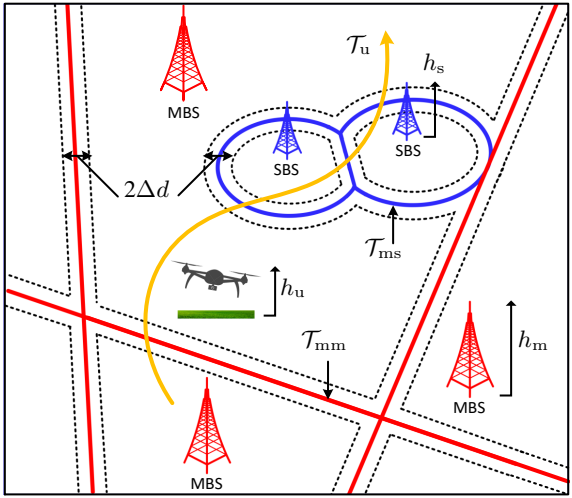
\vspace{-0.25cm}
\section{Handover Rates}
Without loss of generality, we compute the inter and intra-tier HO rates by considering any two tiers, say, $k,j\in\{\mathrm{m,s}\}$ where `m' and `s' denote macro BS (MBS) and small BS (SBS), respectively. Let $\mathcal{T}_{kj}$ be the set of cell boundaries between the $k$-th and the $j$-th tier BSs formed by the virtue of weighted voronoi tessellation (see Fig. \ref{boundaries}). We conduct our analysis on a test user that follows an arbitrary long trajectory $\mathcal{T}_{\mathrm{u}}$ and performs a type $kj$ HO (from $k$-th to $j$-th tier) as soon as it crosses the $kj$ cell boundary. Let $H_{kj}$ be the HO rate per unit time experienced by the test user along its trajectory, which depends on the number of $kj$ boundary crossings per unit trajectory length and the user velocity. In order to determine the HO rate $H_{kj}$, we need to calculate the total number of intersections $N_{kj}$ between $\mathcal{T}_{\mathrm{u}}$ and $\mathcal{T}_{kj}$. Note that the number of intersections between $\mathcal{T}_{\mathrm{u}}$ and $\mathcal{T}_{kj}$ is identical to that of intersections between $\mathcal{T}_{\mathrm{u}}$ and $\mathcal{T}_{jk}$, i.e., $N_{kj}=N_{jk}$. It is worth stating that the number of intersections between the user trajectory and the cell boundaries can be quantified by determining the intensities of cell boundaries. Let $\mu(\mathcal{T}_{kj})$ denote the length intensity of $kj$ cell boundaries, which is the expected length of $kj$ cell boundaries in a unit square. From \cite{20a}, we have a general HO rate expression as a function of $\mu(\mathcal{T}_{kj})$, which is given by
\begin{align}
H_{kj}=\begin{cases}
\frac{2}{\pi}\mu(\mathcal{T}_{kj})v,\text{ for } k= j\\
\frac{1}{\pi}\mu(\mathcal{T}_{kj})v,\text{ for } k\neq j
\end{cases}
\label{Hkj}
\end{align}
\normalsize
\noindent where $\frac{.}{\pi}\mu(\mathcal{T}_{kj})$ denotes the HO rate per unit length (HOL) and $v$ represents the user velocity. Thus the total HO rate is given by
\vspace{-0.2cm}
\begin{align}
H_{\rm Total}=\sum_{k}\sum_{j}H_{kj}.
\end{align}
It is evident from \eqref{Hkj} that the length intensity of cell boundaries is required to compute the HO rates, which is obtained by determining the probability of having the test user on the cell boundary. Note that higher intensity of cell boundaries leads to higher HO rates. Since it is difficult to conduct the analysis on the boundary line, we follow~\cite{HOben} and extend the cells' boundaries by an infinitesimal width $\Delta d$, as illustrated in Fig. \ref{boundaries}. Note that the probability of having the test user on the $\Delta d-$extended cell boundary is equivalent to the expected area of the boundary in a unit square, which can be termed as the area intensity $\mu(\mathcal{T}_{kj}^{(\Delta d)})$. Once we calculate the area intensity, we can then determine the length intensity by letting $\Delta d\rightarrow 0$, i.e.,
\vspace{-0.2cm}
\begin{align}
\mu(\mathcal{T}_{kj})=\lim_{\Delta d\rightarrow 0}\frac{\mu(\mathcal{T}_{kj}^{(\Delta d)})}{2\Delta d}.
\label{length}
\end{align}
\normalsize
It is worth stating that the characterization of the area intensity of cell boundaries involves determining the user-to-BS association probabilities and the service distance distributions. These probabilities and distributions depend on the relative difference between the user and the BSs' (MBS and SBS) antenna heights, which are not considered in the existing literature. \\
In what follows, we calculate the association probabilities and the probability density functions (PDFs) of the distances between the user and the serving BS, which are then exploited to determine the area intensity and finally the HO rate. Note that the mathematical analysis in this paper follows the two-dimensional approach parameterized with network elements' heights.
\vspace{-0.4cm}
\subsection{Association Probabilities}
Let $Z_{\mathrm{m}}$ and $Z_{\mathrm{s}}$ be the Euclidean distance between a test user and the closest MBS and SBS, respectively. The user located in a 3-dimensional plane with height $h_\mathrm{u}$ associates to the MBS if it provides the highest biased received signal strength (RSS)\footnote{As in \cite{HOben,HO_PCP,diff_pathloss,5a,tractable}, we consider a widely accepted maximum biased received power (BRP) based association rule that does not depend on the BS antenna characteristics e.g. aperture and radiation pattern. Also, we assume that the BS antennas are properly designed to cover a wide range of user heights.} i.e., $B_\mathrm{m}P_\mathrm{m} Z_{\mathrm{m}}^{-\eta}>B_\mathrm{s}P_\mathrm{s} Z_{\mathrm{s}}^{-\eta}$. The association probabilities in a 3-dimensional scenario are expressed in the following Lemma.
\begin{lem}
The association probabilities in a PPP based two tier UDN with the user height $h_\mathrm{u}$ are given by
\vspace{-0.2cm}
\begin{multline}
 A_k=\underbrace{1-e^{\pi\lambda_\mathrm{m} \left[h_{\mathrm{um}}^{2} -\beta_{\mathrm{m}j}h_{\mathrm{u}j}^{2}\right]}}_\text{$A_{k_{1}}$}+\\[-0.5em]
 \underbrace{\lambda_{kj}e^{\pi\lambda_\mathrm{m} \left[h_{\mathrm{um}}^{2}-\beta_{\mathrm{ms}}h_{\mathrm{us}}^{2}\right]}}_\text{$A_{k_{2}}$},\text{ for }h_{\mathrm{um}}\leq h_{\mathrm{us}}
 \label{A1}
\end{multline}
\vspace{-0.5cm}
\begin{multline}
 A_k=\underbrace{1-e^{\pi\lambda_\mathrm{s} \left[h_{\mathrm{us}}^{2} -\beta_{\mathrm{s}j}h_{\mathrm{u}j}^{2}\right]}}_\text{$A_{k_{1}}$}+\\[-0.5em]
 \underbrace{\lambda_{kj}e^{\pi\lambda_\mathrm{s} \left[h_{\mathrm{us}}^{2}-\beta_{\mathrm{sm}}h_{\mathrm{um}}^{2}\right]}}_\text{$A_{k_{2}}$},\text{ for }h_{\mathrm{um}}> h_{\mathrm{us}},
 \label{A2}
\end{multline}
where \eqref{A1} and \eqref{A2} hold for $k,j\in\{\mathrm{m,s}\}, k\neq j$ with $\lambda_{kj}= \frac{\lambda_k}{\lambda_k+\lambda_j\beta_{jk}}$, $\beta_{kj}=\frac{1}{\beta_{jk}}=(\frac{B_k P_k}{B_j P_j})^{2/\eta}$, $h_{\mathrm{um}}=\vert h_\mathrm{u}-h_\mathrm{m}\vert$, and $h_{\mathrm{us}}=\vert h_\mathrm{u}-h_\mathrm{s}\vert$.
\label{association}
\end{lem}
\begin{proof}
See Appendix A.
\end{proof}
\vspace{-0.3cm}
\subsection{Service Distance Distribution}
In this section, we calculate the distributions of the distances between the test user and the serving macro and small BSs, which are subsequently used to obtain the area intensities. Note that the distance distributions are characterized here by ordering the BSs w.r.t. to their heights. The service distance distributions are given by the following Lemma.

\begin{lem}
Let $X_{k}$, $k\in\{\mathrm{m,s}\}$, be the horizontal distance between the test user and the serving BS given that the association is with the $k$-th tier BS. Then the PDFs of the horizontal distances between the test user and the serving macro and small BSs are given by
\begin{equation}
    f_{X_{\mathrm{m}}}(x) = \hspace{-0.1cm}\begin{cases}
        \frac{2\pi\lambda_\mathrm{m}}{A_{\mathrm{m}_{1}}} x e^{-\pi\lambda_\mathrm{m} x^2},\text{ for } 0\leq x\leq L_{\mathrm{m}}\\
        \frac{2\pi\lambda_\mathrm{m}}{A_{\mathrm{m}_{2}}} x e^{-\pi x^2\left(\lambda_\mathrm{m}+\lambda_\mathrm{s}P\beta_{\mathrm{sm}}\right)-\pi\lambda_\mathrm{s}( h_{\mathrm{um}}^{2}\beta_{\mathrm{sm}}- h_{\mathrm{us}}^{2})},\\[-0.3em] \hspace{2.5cm}\text{for }L_{\mathrm{m}} \leq x\leq \infty
        \end{cases}
        \label{macrodist}
  \end{equation}
  \vspace{-0.25cm}
  \begin{equation}
    f_{X_{\mathrm{s}}}(x) = \hspace{-0.1cm}\begin{cases}
        \frac{2\pi\lambda_\mathrm{s}}{A_{\mathrm{s}_{1}}} x e^{-\pi\lambda_\mathrm{s} x^2},\text{ for } 0\leq x\leq L_{\mathrm{s}}\\
        \frac{2\pi\lambda_\mathrm{s}}{A_{\mathrm{s}_{2}}} x e^{-\pi x^2\left(\lambda_\mathrm{s}+\lambda_mP_{ms}\right)-\pi\lambda_\mathrm{m}( h_{\mathrm{us}}^{2}\beta_\mathrm{ms}- h_{\mathrm{um}}^{2})},\\[-0.3em] \hspace{2.5cm}\text{for }L_{\mathrm{s}} \leq x\leq \infty
        \end{cases}
        \label{smalldist}
 \end{equation}
  where $f_{X_{\mathrm{m}}}(x)$ and $f_{X_{\mathrm{s}}}(x)$ represent the service distance distributions for the MBS and SBS cases, respectively, while $L_\mathrm{m}$ and $L_\mathrm{s}$ are given by
  \begin{equation*}
  L_{\mathrm{m}}=\begin{cases}
  \sqrt{h_{\mathrm{us}}^2 \beta_\mathrm{ms}-h_{\mathrm{um}}^{2}}, \text{ for } h_{\mathrm{um}}\leq h_{\mathrm{us}}\\
  0\hspace{2.4cm}, \text{ otherwise }\\
  \end{cases}
  \end{equation*}
  \vspace{-0.2cm}
  \begin{equation*}
  L_{\mathrm{s}}=\begin{cases}
  \sqrt{h_{\mathrm{um}}^2 \beta_\mathrm{sm}-h_{\mathrm{us}}^{2}}, \text{ for } h_{\mathrm{um}}> h_{\mathrm{us}}\\
  0\hspace{2.4cm}, \text{ otherwise }\\
  \end{cases}
  \end{equation*}
\end{lem}
\begin{proof}
See Appendix B.
\end{proof}
Since we have computed the service distance distributions, we can now characterize the area intensity of the cell boundaries, which is required to calculate the HO rates. Note that the area intensity of cell boundary refers to the probability of having any arbitrary point on the extended cell boundary in a unit square.
\vspace{-0.2cm}
\subsection{Area Intensities}
In this section, we calculate the area intensity of the cell boundaries. As mentioned earlier, $\mathcal{T}_{kj}$ corresponds to the set of points where the biased received power from the two neighboring BSs is same. Let us assume that the test user located at \textbf{0}$=(0,0,h_\mathrm{u})$ is connected to the MBS located at $(r_0,0,h_\mathrm{m})$. Let $(x_0,y_0,h_\mathrm{s})$ denote the position of a neighboring SBS that provides the best biased RSS among all SBSs. Then $\mathcal{T}_{kj}$ could be any point $(x,y)$ on the trace of cell boundary between macro and small BSs, which can be defined as
\vspace{-0.3cm}
\begin{multline}
\mathcal{T}_{kj}=\Bigg\{(x,y) \bigg | \frac{B_\mathrm{m}P_\mathrm{m}}{[(x-r_0)^2+y^2+h_{\mathrm{um}}^2]^{\eta/2}}=\\
\frac{B_\mathrm{s}P_\mathrm{s}}{[(x-x_0)^2+(y-y_0)^2+h_{\mathrm{us}}^2]^{\eta/2}}\Bigg\}.
\label{Tkj}
\end{multline}
\normalsize
We can now extend the cell boundaries by $\Delta d$, which leads to the extended cell boundaries defined as
\begin{align}
\mathcal{T}_{kj}^{(\Delta d)}=\left\{\mathbf{u}\vert \exists \mathbf{v}\in \mathcal{T}_{kj}, \text{ s.t. } \vert \mathbf{u}-\mathbf{v} \vert <\Delta d\right\}.
\end{align}
\normalsize
Now, we calculate the probability of having the test user on the extended cell boundary given that the user is connected to the tier-$k$ BS, which is given by the following Lemma.
\begin{lem}
\label{th1}
Let a test user located at $\mathbf{0}$ lie on $\mathcal{T}_{kj}^{(\Delta d)}$ while being served by the tier-$k$ BS located at $r_k$, then the conditional probability that $\mathbf{0}\in \mathcal{T}_{kj}^{(\Delta d)}$ is given by
\vspace{-0.1cm}
\begin{eqnarray}
\mathbb{P}[\mathbf{0} \!\in\!\mathcal{T}_{kj}^{(\Delta d)}\vert X\!=\!r_k,n\!=\!k]\!=\!2\lambda_j\Delta d \vartheta(\alpha_{kj},r_k)\!+\!\mathcal{O}(\Delta d^2)
\label{th1eq}
\end{eqnarray}
\normalsize
where $n\in\{\mathrm{m,s}\}$ represents the associated tier, $\mathcal{O}(.)$ denotes the Big O function, and $\vartheta(\alpha_{kj},r_k)$ is given in \eqref{var}.
\begin{figure*}[!t]
\begin{align}
\vartheta(\alpha_{kj},r_k)=\begin{cases}
\frac{1}{\beta_{kj}}\bigints\limits_{0}^{\pi}\sqrt{r_k^2(\beta_{kj}+1)+h_{\mathrm{u}k}^{2}\beta_{kj}-h_{\mathrm{u}j}^{2}\beta_{kj}^{2}-2\beta_{kj}r_k\cos\theta\sqrt{\frac{r_k^2+h_{\mathrm{u}k}^{2}}{\beta_{kj}}-h_{\mathrm{u}j}^{2}}}d\theta \text{ for } k\neq j\\
x\int_{0}^{\pi}\sqrt{2-2\cos{\theta}}d\theta=4x \text{ for } k=j
\label{var}
\end{cases}
\end{align}
\normalsize
\end{figure*}
\begin{proof}
See Appendix C.
\end{proof}
\label{lem3}
\end{lem}
\begin{table}[!t]
\caption{\: Simulation parameters in accordance with \cite{3GPP_R14}}
\center
\resizebox{0.4\textwidth}{!}{
\begin{tabular}{|c c c c|}
\hline
\rowcolor{purple}
\multirow{-1}{*}{\textcolor{white}{\textbf{Parameter}}} & \multirow{-1}{*}{\textcolor{white}{\textbf{Value}} } &\multirow{-1}{*}{ \textcolor{white}{ \textbf{Parameter}}}  & \multirow{-1}{*}{\textcolor{white}{ \textbf{Value }}}  \\ \hline  \hline
& & & \\
 \multirow{-2}{*}{MBS Power $P_\mathrm{m}$:}             &  \multirow{-2}{*}{46 dBm }    &  \multirow{-2}{*}{SBS Power $P_\mathrm{s}$:}  &   \multirow{-2}{*}{24 dBm}     \\
 \cellcolor{purple!20!}  & \cellcolor{purple!20!}    &\cellcolor{purple!20!}    & \cellcolor{purple!20!}  \\
\multirow{-2}{*}{\cellcolor{purple!20!} MBS $\lambda_\mathrm{m}$:}  & \multirow{-2}{*}{\cellcolor{purple!20!}3 BS/km$^2$} & \multirow{-2}{*}{\cellcolor{purple!20!}SBS intensity $\lambda_\mathrm{s}$:}  &   \multirow{-2}{*}{\cellcolor{purple!20!} 10 BS/km$^2$}   \\
& & & \\
 \multirow{-2}{*}{MBS height $h_\mathrm{m}$:}               & \multirow{-2}{*}{40 m}  &  \multirow{-2}{*}{SBS height $h_\mathrm{s}$:}       &  \multirow{-2}{*}{25 m}   \\
\cellcolor{purple!20!} & \cellcolor{purple!20!}  & \cellcolor{purple!20!}  &\cellcolor{purple!20!}   \\
 \multirow{-2}{*}{\cellcolor{purple!20!}User velocity $v$:}           &  \multirow{-2}{*}{\cellcolor{purple!20!} 30 km/h}     & \multirow{-2}{*}{\cellcolor{purple!20!} Path loss exponent $\eta$:}     &  \multirow{-2}{*}{\cellcolor{purple!20!} 4}
  \\ \hline
\end{tabular}
}
\label{tab}
\end{table}
The special case of having the user and the BSs antennas at the same height reduces Lemma \ref{lem3} into much simpler expression as shown in the next corollary.
\vspace{0.01cm}
\begin{col}

In the special case of $h_{\mathrm{um}}=h_{\mathrm{us}}=0$, $\mathbb{P}[\mathbf{0} \in\mathcal{T}_{kj}^{(\Delta d)}\vert X=r_k,n=k]$ is given by
\vspace{-0.2cm}
\begin{multline}
\mathbb{P}[\mathbf{0} \in\mathcal{T}_{kj}^{(\Delta d)}\vert X\!=\!r_k,n\!=\!k]=\\
\frac{2\lambda_j\Delta d r_k}{\beta_{kj}}\int\limits_{0}^{\pi}\sqrt{\beta_{kj}+1-2\sqrt{\beta_{kj}}\cos\theta}d\theta+\mathcal{O}(\Delta d^2).
\label{special}
\end{multline}
\normalsize
\end{col}
Note that \eqref{special} corresponds to the two tier case in~\cite{HOben}.
Since we have calculated the probability that $\mathbf{0}\in \mathcal{T}_{kj}^{(\Delta d)}$ conditioned on the current association (Lemma \ref{th1}), we can now determine the area intensity $\mu(\mathcal{T}_{kj}^{(\Delta d)})$ of the $\Delta d-$extended cell boundaries $\mathcal{T}_{kj}^{(\Delta d)}$, which is given in the next theorem.

\begin{theorem}
The area intensity or the probability of having any arbitrary point on the $\Delta d-$extended inter-tier cell boundaries is given in \eqref{area1}\footnote{As in \cite{lower_height}, it is difficult to find a closed form solution for height-aware models. Therefore, numerical evaluation is performed to solve \eqref{area1}.}.

\begin{figure*}[!t]
\vspace{-0.45cm}
\begin{align}
\mu(\mathcal{T}_{kj}^{(\Delta d)})=\begin{cases}
 \int\limits_{\sqrt{h_{\mathrm{us}}^2 \beta_\mathrm{ms}-h_{\mathrm{um}}^{2}}}^{\infty}2\lambda_j\Delta d \vartheta(\alpha_{kj},x)f_{X_{\mathrm{m}}}(x)dx+\int\limits_{0}^{\infty}2\lambda_k\Delta d \vartheta(\alpha_{jk},x)f_{X_{\mathrm{s}}}(x)dx+ \mathcal{O}(\Delta d^2), \text{ for } h_{\mathrm{um}}\leq h_{\mathrm{us}}\\
\int\limits_{0}^{\infty}2\lambda_j\Delta d \vartheta(\alpha_{kj},x)f_{X_{\mathrm{m}}}(x)dx+\int\limits_{\sqrt{h_{\mathrm{um}}^2 \beta_\mathrm{sm}-h_{\mathrm{us}}^{2}}}^{\infty}2\lambda_k\Delta d \vartheta(\alpha_{jk},x)f_{X_{\mathrm{s}}}(x)dx+ \mathcal{O}(\Delta d^2), \text{ for } h_{\mathrm{um}}> h_{\mathrm{us}}
\label{area1}
\end{cases}
\end{align}
\vspace{-0.25cm}
\begin{align}
\mu(\mathcal{T}_{kk}^{(\Delta d)})=\int_{0}^{L_k}2\lambda_k\Delta d\vartheta(\alpha_{kk},x)f_{X_{k}}(x)dx+\int_{L_k}^{\infty}2\lambda_k\Delta d\vartheta(\alpha_{kk},x)f_{X_{k}}(x)dx
\label{area2}
\end{align}
\hrulefill
\normalsize
\end{figure*}
\begin{proof}
See Appendix D.
\end{proof}
\end{theorem}
\begin{col}
The area intensity of $\Delta d-$extended intra-tier cell boundaries can be obtained by setting $k=j$ in \eqref{area1} and is given in \eqref{area2}.
\label{col2}
\end{col}
Note that \eqref{area2} refers to the intra-tier HO scenario where the user height is different from that of serving/target BS. The area intensities of the cell boundaries computed in Theorem 1 and Corollary \ref{col2} can now be exploited to determine the inter and intra-tier HO rates, respectively. First, we calculate the length intensity from the area intensity as shown in \eqref{length} and then substitute the result in \eqref{Hkj} to obtain the HO rates.
\begin{figure}[!t]
\centering
\includegraphics[width=0.9\linewidth,keepaspectratio]{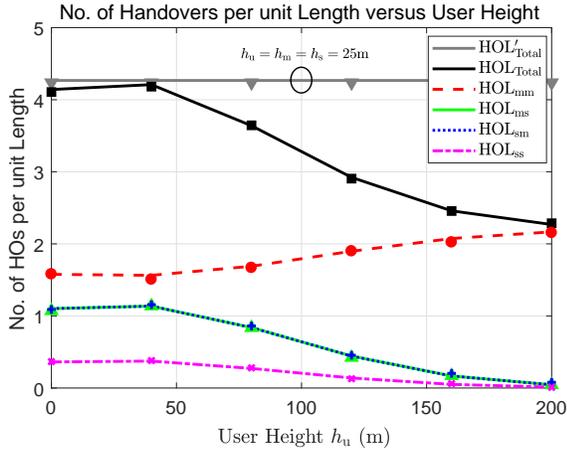} \vspace{-0.2cm}\caption{Handover rates per unit length versus user height without biasing}
\label{HO1}
\end{figure}

\begin{figure}[!t]
\centering
\vspace{-0.2cm}
\includegraphics[width=0.9\linewidth,keepaspectratio]{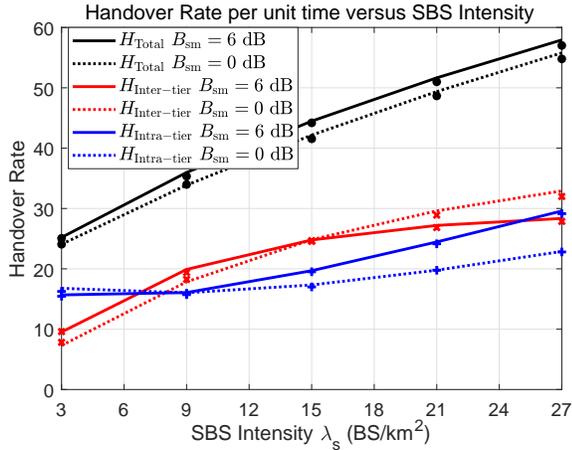} \vspace{-0.2cm}\caption{Handover rates versus SBS intensity with $B_\mathrm{sm}=\frac{B_\mathrm{s}}{B_\mathrm{m}}$}
\label{HO2}
\vspace{-0.25cm}
\end{figure}
\section{Numerical Results}
\label{results}
Figs. \ref{HO1} and \ref{HO2} show the HO rates with lines and markers representing analysis and simulations results, respectively, with parameters shown in Table \ref{tab}. Fig. \ref{HO1} depicts the HO rates for constant and varying user heights. Note that $\mathrm{HOL}_{\mathrm{Total}}^{\prime}$ refers to the conventional total HO rate where there exists no disparity among the user and BSs heights. Despite the decreasing overall HO rate with the increasing user height, Fig. \ref{HO1} marks MBSs as bearing the increasing HO trend. This is because we assume that $h_{\mathrm{m}} >h_{\mathrm{s}}$ and $P_{\mathrm{m}} >P_{\mathrm{s}}$, which is the practical scenario. Thus there arises a need to provide an additional signalling capacity to facilitate macro-to-macro HOs in the scenarios where the relative difference between the user and the BSs' antenna heights is not negligible. However, biasing can be employed to balance the HO signalling load between the two tiers. Fig. \ref{HO2} epitomizes the HO rates with varying SBS intensity and the biasing factors of 0 dB and 6 dB. It is evident from the trends that the 6 dB bias in SBSs results in balancing the inter and intra-tier HO rates.

\section{Conclusion}
\label{conclusion}
This article characterizes the association probabilities and inter/intra tier HO rates in $N$-tier UDN by incorporating user/BSs antenna heights into the mathematical analysis. The proposed analytical model can be applied to various practical scenarios where the user could be a land vehicle or an aerial vehicle. In particular, we study the impact of user height on the association probabilities and HO rate, and validate our model via Monte Carlo simulations. The numerical results show decrease in the overall HO rate with the increase in the user height while giving rise to the macro-to-macro HO rate, which can be balanced by adding a positive bias to the SBSs. By the virtue of this model, some HO management techniques could be investigated to minimize the effect of HO rate on the UDN KPIs. As a future work, we will study more advanced propagation and fading models for user-to-BS associations and HO rates.
\vspace{-0.4cm}
\appendices
\section{Proof of Lemma 1}
Let $R_\mathrm{m}=\sqrt{Z_{\mathrm{m}}^{2}-h_{\mathrm{um}}^{2}}$ be the horizontal distance between the MBS and the test user. Then the distribution of $Z_m$ can be calculated as
\vspace{-0.2cm}
\begin{align}
F_{Z_{\mathrm{m}}}(x)&=\mathbb{P}[R_{\mathrm{m}} \leq \sqrt{x^2 -h_{\mathrm{um}}^{2}} \hspace{0.05cm}] \notag \\
&\myeq 1-e^{-\pi\lambda_\mathrm{m} (x^2 - h_{\mathrm{um}}^2)}, \quad  h_{\mathrm{um}}<x<\infty,
\end{align}\normalsize
 where (a) follows from the null probability of PPP. Similarly, the distribution of $Z_s$ is given by
 \begin{align}
 F_{Z_{\mathrm{s}}}(x)=1-e^{-\pi\lambda_\mathrm{s} (x^2 - h_{\mathrm{us}}^2)},\quad h_{\mathrm{us}}<x<\infty,
\end{align}\normalsize
  For the macro association probability $A_\mathrm{m}$, we first write
  \begin{align}
  A_\mathrm{m}&=\mathbb{P}[B_\mathrm{m}P_\mathrm{m} Z_{\mathrm{m}}^{-\eta}>B_\mathrm{s}P_\mathrm{s} Z_{\mathrm{s}}^{-\eta}]\notag \\
  &= \mathbb{E}_{Z_{\mathrm{m}}}\bigg\{\mathbb{P}[Z_{\mathrm{s}}>\bigg(\frac{B_\mathrm{s}P_\mathrm{s}}{B_\mathrm{m}P_\mathrm{m}}\bigg)^{1/\eta} Z_{\mathrm{m}}]\bigg\}
  \label{Ai}
   \end{align}
   \normalsize
   Then we solve \eqref{Ai} by exploiting the fact that $P[Z_{\mathrm{s}}>(\frac{B_\mathrm{s}P_\mathrm{s}}{B_\mathrm{m}P_\mathrm{m}})^{1/\eta} Z_{\mathrm{m}}]=1$ over the range $h_{\mathrm{um}}\leq Z_m\leq h_{\mathrm{us}}\big(\frac{B_\mathrm{m}P_\mathrm{m}}{B_\mathrm{s}P_\mathrm{s}}\big)^{1/\eta}$ and $P[Z_{\mathrm{s}}>(\frac{P_\mathrm{s}}{P_\mathrm{m}})^{1/\eta} Z_{\mathrm{m}}]=e^{-\pi\lambda_\mathrm{s}[(\frac{B_\mathrm{s}P_\mathrm{s}}{B_\mathrm{m}P_\mathrm{m}})^{2/\eta}Z_{\mathrm{m}}^{2}-h_{\mathrm{us}}^2]}$ over the range $h_{\mathrm{us}}\big(\frac{B_\mathrm{m}P_\mathrm{m}}{B_\mathrm{s}P_\mathrm{s}}\big)^{1/\eta}<Z_m<\infty$ given that $h_{\mathrm{um}}\leq h_{\mathrm{us}}$. In case of $h_{\mathrm{um}}> h_{\mathrm{us}}$, then $P[Z_{\mathrm{s}}>(\frac{B_\mathrm{s}P_\mathrm{s}}{B_\mathrm{m}P_\mathrm{m}})^{1/\eta} Z_{\mathrm{m}}]=e^{-\pi\lambda_\mathrm{s}[(\frac{B_\mathrm{s}P_\mathrm{s}}{B_\mathrm{m}P_\mathrm{m}})^{2/\eta}Z_{\mathrm{m}}^{2}-h_{\mathrm{us}}^2]}$ over the range $h_{\mathrm{um}}<Z_m<\infty$. A similar approach is followed to calculate the SBS association probability $A_\mathrm{s}$.
   \vspace{-0.25cm}
\section{Proof of Lemma 2}
In order to calculate the distance distribution, we first calculate the complementary cumulative distribution function (CCDF) given that the user associates with the $k$-th tier, $k\in\{\mathrm{m,s}\}$.
\vspace{-0.2cm}
\begin{eqnarray}
\mathbb{P}[X_{k}> x]=\mathbb{P}[R_{k}> x\big| n=k]=\frac{\mathbb{P}[R_{k}> x, n=k]}{\mathbb{P}[n=k]}
\label{A}
\end{eqnarray}
where $\mathbb{P}[n=k]=A_{k}$, which is given in Lemma~\ref{association}. The joint distribution $\mathbb{P}[R_{k}> x, n=k]$ is calculated as follows
\begin{align}
\!\!\!\!\mathbb{P}[R_{k}\!>\! x, n\!=\!k]\!&=\!\mathbb{P}[R_{k}\!> x, B_k P_k Z_{k}^{-\eta}\!>\!B_j P_j Z_{j}^{-\eta}] \notag\\
&\myeqb\int_{x}^{\infty}e^{-\pi\lambda_j[\beta_{jk}(x^2+h_{\mathrm{u}k}^{2})-h_{\mathrm{u}j}^{2}]}f_{R_{k}}(x)dx
\label{cond}
\end{align}
where (b) follows from the null probability of PPP with $f_{R_{k}}(x)=2\pi\lambda_k x e^{-\pi\lambda_kx^2}$. Now
$f_{X_{k}}(x)$ is calculated by substituting \eqref{cond} in \eqref{A} and invoking $A_k$ given in \eqref{A1} and \eqref{A2} conditioned on the heights i.e., $h_{\mathrm{um}}$ and $h_{\mathrm{us}}$, and then taking the derivative w.r.t. $x$.
\vspace{-0.25cm}
\section{Proof of Lemma 3}
Let $\mathcal{S}$ be the location of tier-$j$ SBS such that the condition \textbf{0} $\in \mathcal{T}_{kj}^{(\Delta d)}$ is satisfied. Mathematically, we can define $\mathcal{S}$ as
\vspace{-0.15cm}
\begin{align}
\mathcal{S}=\left\{x_0,y_0\vert \mathbf{d}<\Delta d\right\}
\label{S}
\end{align}
\normalsize
where $\mathbf{d}$ is the distance between \textbf{0} and $\mathcal{T}_{kj}^{(\Delta d)}$, which is calculated by the mathematical manipulation of \eqref{Tkj}.
It is worth mentioning that $\mathcal{T}_{kj}^{(\Delta d)}$ represents a circle centered at $[\frac{\beta_{\mathrm{ms}} x_0-r_0}{\beta_{\mathrm{ms}}-1},\frac{\beta_{\mathrm{ms}}y_0}{\beta_{\mathrm{ms}}-1}]$ with radius $\frac{\sqrt{\beta_{\mathrm{ms}}(r_0^2+x_0^2+y_0^2-2x_0r_0+h_{\mathrm{um}}^{2}+h_{\mathrm{us}}^{2}-\frac{h_{\mathrm{um}}^{2}}{\beta_{\mathrm{ms}}}-h_{\mathrm{us}}^{2}\beta_{\mathrm{ms}})}}{\beta_{\mathrm{ms}}-1}$. Thus the distance from \textbf{0} to the trace is given by

\vspace{-0.2cm}
\small
\begin{multline}
\mathbf{d}=\frac{\sqrt{(\beta_{\mathrm{ms}}x_0-r_0)^2+\beta_{\mathrm{ms}}^{2}y_{0}^{2}}}{\beta_{\mathrm{ms}}-1}-\\ \!\!\!\!\!\frac{\sqrt{\beta_{\mathrm{ms}}(r_0^2+x_0^2+y_0^2-2x_0r_0+h_{\mathrm{um}}^{2}+h_{\mathrm{us}}^{2}-\frac{h_{\mathrm{um}}^{2}}{\beta_{\mathrm{ms}}}-h_{\mathrm{us}}^{2}\beta_{\mathrm{ms}})}}{\beta_{\mathrm{ms}}-1}
\label{distance}
\end{multline}
\normalsize
Substituting \eqref{distance} in \eqref{S} and transforming $(x_0,y_0)$ to the polar coordinates $(r,\theta)$ yields

\vspace{-0.2cm}
\small
\begin{multline*}
\mathcal{S}=\bigg\{(r,\theta)\bigg\vert \bigg\vert r^2-\frac{r_{0}^{2}}{\beta_{\mathrm{ms}}}-\frac{h_{\mathrm{um}}^{2}}{\beta_{\mathrm{ms}}}+h_{\mathrm{us}}^{2}\bigg\vert <\frac{2\Delta d}{\beta_{\mathrm{ms}}} \Bigg[r_{0}^{2}(1+\!\beta_{\mathrm{ms}})+\\ h_{\mathrm{um}}^{2}\beta_{\mathrm{ms}}-
h_{\mathrm{us}}^{2}\beta_{\mathrm{ms}}^{2}-2\beta_{\mathrm{ms}}r_0\cos(\theta)\sqrt{\frac{r_{0}^{2}+h_{\mathrm{um}}^{2}}{\beta_{\mathrm{ms}}}\!-\!h_{\mathrm{us}}^{2}}+\mathcal{O}(\Delta d^2)\Bigg]^{1/2}
\end{multline*}
\normalsize
It is worth stating that there is no tier-$j$ BS between \textbf{0} and $\sqrt{\frac{r_{0}^{2}+h_{\mathrm{um}}^{2}}{\beta_{\mathrm{ms}}}\!-\!h_{\mathrm{us}}^{2}}$. This implies that $\mathcal{S}$ represents a ring region with area given by
\begin{multline*}
\mathcal{S}_A=\frac{2\Delta d}{\beta_{\mathrm{ms}}}
\int_{0}^{\pi}\bigg[r_{0}^{2}(\!1+\!\!\beta_{\mathrm{ms}})+h_{\mathrm{um}}^{2}\beta_{\mathrm{ms}}-\\
h_{\mathrm{us}}^{2}\beta_{\mathrm{ms}}^{2}-2\beta_{\mathrm{ms}}r_0\cos(\theta)\sqrt{\frac{r_{0}^{2}+h_{\mathrm{um}}^{2}}{\beta_{\mathrm{ms}}}-h_{\mathrm{us}}^{2}}\Bigg]^{1/2}d\theta +\mathcal{O}(\Delta d^2)
\end{multline*}
\normalsize
where $\mathcal{S}_A$ can be written as
\begin{align}
\mathcal{S}_A=2\Delta d \vartheta(\alpha_{kj},r_k) +\mathcal{O}(\Delta d^2)
\end{align}
\normalsize
Now $\mathbb{P}[\mathbf{0} \in\mathcal{T}_{kj}^{(\Delta d)}\vert X=r_k,n=k]$ is calculated using the fact that there is at least one tier-$j$ BS in $\mathcal{S}$, which is given by

\vspace{-0.2cm}
\small
\begin{align}
\!\!\!\!\mathbb{P}[\mathbf{0} \in\mathcal{T}_{kj}^{(\Delta d)}\vert X\!=\!r_k,n\!=\!k]&=1-e^{- \lambda_j\mathcal{S}_A} \notag \\
&=2\lambda_j\Delta d \vartheta(\alpha_{kj},r_k)+\mathcal{O}(\Delta d^2)
\end{align}
\normalsize
\vspace{-0.25cm}
\section{Proof of Theorem 1}
Let $\mu(\mathcal{T}_{kj}^{(\Delta d)})$ be the area intensity of cell boundaries, which is equal to the probability of having an arbitrary point in a unit square. This implies that

\vspace{-0.2cm}
\small
\begin{multline}
\mu(\mathcal{T}_{kj}^{(\Delta d)})=\mathbb{P}[\mathbf{0}\in\mathcal{T}_{kj}^{(\Delta d)}]\\
=\!\!\!\!\!\!\!\int\limits_{r_{\mathrm{m}}\in S_m}^{}\!\!\!\mathbb{P}[\mathbf{0} \in\mathcal{T}_{kj}^{(\Delta d)}\vert X=r_\mathrm{m},n=k]
f_X(r_{\mathrm{m}}\vert n=k)\mathbb{P}[n=k] dr_{\mathrm{m}}\\+\!\!\!\!\!\!\! \int\limits_{r_{\mathrm{s}}\in S_s}^{}\mathbb{P}[\textbf{0} \in\mathcal{T}_{kj}^{(\Delta d)}\vert X=r_{\mathrm{s}},n=j]
f_X(r_{\mathrm{s}}\vert n=j)\mathbb{P}[n=j] dr_{\mathrm{s}}
\label{area}
\end{multline}
\normalsize
where the integration regions $S_m$ and $S_s$ correspond to the heights dependent boundaries given in \eqref{macrodist} and \eqref{smalldist}, respectively. The theorem is proved by the substitution of \eqref{A1}, \eqref{A2}, \eqref{macrodist}, \eqref{smalldist}, and \eqref{th1eq} in \eqref{area}.
\vspace{-0.4cm}
\bibliographystyle{IEEEtran}
\bibliography{IEEEabrv,mybibr}

\begin{thebibliography}{10}
\providecommand{\url}[1]{#1}
\csname url@samestyle\endcsname
\providecommand{\newblock}{\relax}
\providecommand{\bibinfo}[2]{#2}
\providecommand{\BIBentrySTDinterwordspacing}{\spaceskip=0pt\relax}
\providecommand{\BIBentryALTinterwordstretchfactor}{4}
\providecommand{\BIBentryALTinterwordspacing}{\spaceskip=\fontdimen2\font plus
\BIBentryALTinterwordstretchfactor\fontdimen3\font minus
  \fontdimen4\font\relax}
\providecommand{\BIBforeignlanguage}[2]{{%
\expandafter\ifx\csname l@#1\endcsname\relax
\typeout{** WARNING: IEEEtran.bst: No hyphenation pattern has been}%
\typeout{** loaded for the language `#1'. Using the pattern for}%
\typeout{** the default language instead.}%
\else
\language=\csname l@#1\endcsname
\fi
#2}}
\providecommand{\BIBdecl}{\relax}
\BIBdecl

\bibitem{5GUDN}
X.~Ge, S.~Tu, G.~Mao, C.-X. Wang, and T.~Han, ``5{G} ultra-dense cellular
  networks,'' \emph{IEEE Wireless Communications}, vol.~23, no.~1, pp. 72--79,
  2016.

\bibitem{sun2017impact}
W.~Sun and Y.~Teng, ``Impact of user mobility on transmit power control in
  ultra dense networks,'' in \emph{IEEE International Conference on
  Communications Workshops (ICC Workshops)}, 2017, pp. 1165--1170.

\bibitem{park2016user}
J.~Park, S.~Y. Jung, S.-L. Kim, M.~Bennis, and M.~Debbah, ``User-centric
  mobility management in ultra-dense cellular networks under spatio-temporal
  dynamics,'' in \emph{IEEE Global Communications Conference (GLOBECOM)}, 2016,
  pp. 1--6.

\bibitem{Lin}
X.~Lin, R.~K. Ganti, P.~J. Fleming, and J.~G. Andrews, ``Towards understanding
  the fundamentals of mobility in cellular networks,'' \emph{{IEEE} Trans.
  Wireless Commun.}, vol.~12, no.~4, pp. 1686--1698, 2013.

\bibitem{HOben}
W.~Bao and B.~Liang, ``Stochastic geometric analysis of user mobility in
  heterogeneous wireless networks,'' \emph{IEEE J. Sel. Areas Commun.},
  vol.~33, no.~10, pp. 2212--2225, 2015.

\bibitem{HO_PCP}
------, ``Handoff rate analysis in heterogeneous wireless networks with poisson
  and poisson cluster patterns,'' in \emph{Proceedings of the 16th ACM
  International Symposium on Mobile Ad Hoc Networking and Computing}, 2015, pp.
  77--86.

\bibitem{diff_pathloss}
Y.~Ren, Y.~Li, and C.~Qi, ``Handover rate analysis for k-tier heterogeneous
  cellular networks with general path-loss exponents,'' \emph{IEEE
  Communications Letters}, vol.~21, no.~8, pp. 1863--1866, 2017.

\bibitem{lower_height}
M.~Ding and D.~L. P{\'e}rez, ``Please lower small cell antenna heights in
  5{G},'' in \emph{IEEE Global Communications Conference (GLOBECOM)}, 2016, pp.
  1--6.

\bibitem{voronoi}
P.~F. Ash and E.~D. Bolker, ``Generalized dirichlet tessellations,''
  \emph{Geometriae Dedicata}, vol.~20, no.~2, pp. 209--243, 1986.

\bibitem{20a}
S.~N. Chiu, D.~Stoyan, W.~Kendall, and J.~Mecke, \emph{Stochastic Geometry and
  its Applications}.\hskip 1em plus 0.5em minus 0.4em\relax John Wiley \& Sons,
  2013.

\bibitem{5a}
H.-S. Jo, Y.~J. Sang, P.~Xia, and J.~G. Andrews, ``Heterogeneous cellular
  networks with flexible cell association: A comprehensive downlink {SINR}
  analysis,'' \emph{{IEEE} Trans. Wireless Commun.}, vol.~11, no.~10, pp.
  3484--3495, 2012.

\bibitem{tractable}
J.~G. Andrews, F.~Baccelli, and R.~K. Ganti, ``A tractable approach to coverage
  and rate in cellular networks,'' \emph{IEEE Trans. Commun.}, vol.~59, no.~11,
  pp. 3122--3134, 2011.

\bibitem{3GPP_R14}
{3GPP TR 36.931 v14.0.0}, ``Radio {F}requency ({RF}) requirements for {LTE}
  pico {Node} (release 12),'' 3GPP TSG RAN, Tech. Rep.

\end{thebibliography}
\end{document}